\documentclass{article}
\usepackage[totalwidth=13.0cm,totalheight=20.0cm]{geometry}

\usepackage{latexsym,amsthm,amsmath,amssymb,url}

\newtheorem{lemma}{Lemma}

\newtheorem{theorem}{Theorem}

\newtheorem{proposition}{Proposition}

\newcommand{\ep}{\varepsilon}
\newcommand{\MrSAA}{\textsc{Max-$r(n)$-Sat-AA}}
\newcommand{\var}{\text{\normalfont var}}

\newtheorem{krule}{Reduction Rule}

\begin{document}

\title{Parameterized Complexity of MaxSat Above Average\footnote{A preliminary version of this paper will appear in the proceedings of LATIN 2012}}

\author{Robert Crowston\thanks{Royal Holloway, University of London, Egham,
Surrey, UK} \and Gregory Gutin\footnotemark[2] \and Mark Jones\footnotemark[2] \and Venkatesh Raman\thanks{The Institute of Mathematical Sciences,
Chennai 600 113, India} \and Saket Saurabh\footnotemark[3]}

\date{}
\maketitle

\begin{center}
{\bf This paper is dedicated to the memory of Alan Turing}
\end{center}

\begin{abstract}
\noindent
In {\sc MaxSat}, we are given a CNF formula $F$ with $n$ variables and $m$ clauses and asked to find a truth assignment satisfying the maximum number of clauses. Let $r_1,\ldots, r_m$ be the number of literals in the clauses of $F$. Then ${\rm asat}(F)=\sum_{i=1}^m  (1-2^{-r_i})$ is the expected number of clauses satisfied by a random truth assignment (the truth values to the variables are distributed uniformly and independently). It is well-known that, in polynomial time, one can find a truth assignment satisfying at least ${\rm asat}(F)$ clauses. In the parameterized problem {\sc MaxSat-AA}, we are to decide whether there is a truth assignment satisfying at least ${\rm asat}(F)+k$ clauses, where $k$ is the (nonnegative) parameter. We prove that {\sc MaxSat-AA} is para-NP-complete
and thus, {\sc MaxSat-AA} is not fixed-parameter tractable unless P$=$NP. This is in sharp contrast to the similar problem {\sc MaxLin2-AA} which was recently proved to be fixed-parameter tractable by Crowston {\em et al.} (FSTTCS 2011).

In fact, we consider a more refined version of {\sc MaxSat-AA}, {\sc Max-$r(n)$-Sat-AA}, where $r_j\le r(n)$ for each $j$. Alon {\em et al.} (SODA 2010) proved that if $r=r(n)$ is a constant, then {\sc Max-$r$-Sat-AA} is fixed-parameter tractable. We prove that  {\sc Max-$r(n)$-Sat-AA} is para-NP-complete for $r(n)=\lceil \log n\rceil.$ We also prove that assuming the exponential time hypothesis, {\sc Max-$r(n)$-Sat-AA} is not in XP already for any $r(n)\ge \log \log n +\phi(n)$, where $\phi(n)$ is any unbounded strictly increasing function.
This lower bound on $r(n)$ cannot be decreased much further as we prove that {\sc Max-$r(n)$-Sat-AA} is (i) in XP for any $r(n)\le \log \log n - \log \log \log n$ and (ii) fixed-parameter tractable for any $r(n)\le \log \log n - \log \log \log n - \phi(n)$, where $\phi(n)$ is any unbounded strictly increasing function. The proof uses some results on {\sc MaxLin2-AA}.
\end{abstract}

\section{Introduction}

{\sc Satisfiability} is a well-known fundamental problem in Computer
Science. Its optimization version (finding the maximum number of clauses that can be satisfied by a truth assignment) and its generalizations (constraint satisfaction problems) are well studied in almost every paradigm of algorithms and complexity including approximation and parameterized complexity. Here we consider the parameterized complexity of a variation of {\sc MaxSat}.

In parameterized complexity, one identifies a natural parameter $k$ in the input and algorithms are designed and analyzed to confine the combinatorial explosion to this parameter, while keeping the rest of the running time to be polynomial in the size of the input. More specifically, the notion of feasibility is {\em fixed-parameter tractability} where one is interested in an algorithm whose running time is $O(f(k)n^c)$, where $f$ is an arbitrary (typically exponential) function of $k$, $n$ is the input size and $c$ is a constant independent of $k$. When the values of $k$ are relatively small, fixed-parameter tractability implies that the problem under consideration is tractable, in a sense. The class of fixed-parameter tractable problems will be denoted by FPT.
When the available running time is replaced by the much more powerful $n^{O(f(k))},$ we obtain the class XP, where each problem is polynomial-time solvable for any fixed value of $k.$ It is well-known that FPT is a proper subset of XP.
More details on parameterized algorithms and complexity are given at the end of this section.

A well-studied parameter in most optimization problems is the size of the solution. In particular, for {\sc MaxSat}, the natural parameterized question is whether a given boolean formula in CNF has an assignment satisfying at least $k$ clauses. Using the (folklore) observation that every CNF formula on $m$ clauses has an assignment satisfying at least $m/2$ clauses (a random assignment will satisfy at least $m/2$ clauses), Mahajan and Raman \cite{MR99} observed that this problem is fixed-parameter tractable. This lower bound, of $m/2$, for the maximum number of clauses, means that the problem is interesting only when $k > m/2$, i.e., when the values of $k$ are relatively large. Hence Mahajan and Raman introduced and showed fixed-parameter tractable, a more natural parameterized question, namely whether the given CNF formula has an assignment satisfying at least $m/2+k$ clauses.

This idea of parameterizing above a (tight) lower bound has been followed up in many directions subsequently. For {\sc MaxSat} alone, better (larger than $m/2$) lower bounds for certain classes of instances (formulas with no pair of unit clauses in conflict, for example) have been proved and the problems parameterized above these bounds have been shown to be fixed-parameter tractable \cite{ipec2010,GutJonYeo11}.
When every clause has $r$ literals, the expected number of clauses that can be satisfied by a (uniformly) random assignment can easily seen to be
$(1-1/2^r)m$, and
Alon {\em et al.} \cite{AloGutKimSzeYeo11} proved that checking whether $k$ more than this many clauses can be satisfied in such an $r$-CNF formula is
fixed-parameter tractable. This problem is known to be \textsc{Max-$r$-Sat-AA}.
The problem \MrSAA{} we consider in this paper, is a refinement of this problem, where $r$ need not be a constant.

More specifically, the problem \textsc{MaxSat-AA}\footnote{In this paper, AA is an abbreviation for Above Average}, we address is the following.
\begin{quote}
\textsc{MaxSat-AA} \\ \nopagebreak
  \emph{Instance:} A CNF formula $F$ with clauses $c_1,\ldots , c_m$, and variables $x_1,\ldots ,x_n$, and a nonnegative integer $k$.
  Clause $c_i$ has $r_i$ literals, $i=1,\ldots ,m$. \\
   \emph{Parameter:} $k$.\\
   \emph{Question:} Decide whether there is a truth assignment satisfying at least ${\rm asat}(F)+k$ clauses, where ${\rm asat}(F)=\sum_{i=1}^m  (1-2^{-r_i}).$
\end{quote}
The problem \MrSAA{} is a refinement of \textsc{MaxSat-AA} in which each clause has at most $r(n)$ literals.

%
%

Note that ${\rm asat}(F)$ is the average number of satisfied clauses. Indeed, if we assign {\sc true} or {\sc false} to each $x_j$ with probability $1/2$ independently of the other variables, then the probability of $c_i$ being satisfied is $1-2^{-r_i},$ and by linearity of expectation, ${\rm asat}(F)$ is the expected number of satisfied clauses. (Since our distribution is uniform, ${\rm asat}(F)$ is indeed the average number of satisfied clauses.) Let ${\rm sat}(F)$ denote the maximum number of clauses satisfied by a truth assignment. For Boolean variables $y_1,\ldots ,y_t$, the {\em complete set of clauses on} $y_1,\ldots ,y_t$ is the set $\{(z_1\vee \ldots \vee z_t) : z_i \in \{y_i,\bar{y_i} \}, i\in [t]\}$.
Any formula $F$ consisting of one or more complete sets of clauses shows that the lower bound ${\rm sat}(F)\ge {\rm asat}(F)$ on ${\rm sat}(F)$ is tight. Using the derandomization method of conditional expectations (see, e.g., \cite{alon}), it is easy to obtain a polynomial time algorithm which finds a truth assignment satisfying at least ${\rm asat}(F)$ clauses. Thus, the question asked in {\sc MaxSat-AA} is whether we can find a better truth assignment efficiently from the parameterized complexity point of view.

%
{\bf New Results.} Solving an open problem of \cite{CroFellowsEtAl11} we show that \textsc{MaxSat-AA} is not fixed-parameter tractable unless P$=$NP. More specifically, we show this for \textsc{Max-$r(n)$-Sat-AA} for
$r(n)=\lceil \log n\rceil.$  Also, we prove that unless the exponential time hypothesis (ETH) is false, \MrSAA{} is not even in XP for any $r(n)\ge \log \log n +\phi(n)$, where $\phi(n)$ is any unbounded strictly increasing function\footnote{A function $f$ is {\em strictly increasing} if for every pair $x',x''$ of values of the argument with $x'<x''$, we have $f(x')<f(x'').$}. These two results are proved in Section \ref{sec:neg}.

These results are in sharp contrast to the related problems \textsc{MaxLin2-AA} (see Section \ref{sec:pos} for a definition of this problem) and \textsc{Max-$r$-Sat-AA} which are known to be fixed-parameter tractable. Also this is one of the very few problems in the `above guarantee' parameterization world, which is known to be hard. See \cite{MahajanRamanSikdar09} for a few other hard above guarantee
problems.

Then, complementing our hardness results, we show that
the lower bound above on $r(n)$ cannot be decreased much further as we prove that \MrSAA{} is in XP for any $r(n)\le \log \log n - \log \log \log n$ and
fixed-parameter tractable for any $r(n)\le \log \log n - \log \log \log n - \phi(n)$, where $\phi(n)$ is any unbounded strictly increasing function. This result generalizes the one of Alon {\em et al.} \cite{AloGutKimSzeYeo11} and is proved in Section \ref{sec:pos}.

The problem we study is one of the few problems in the `above guarantee parameterization' where we parameterize above an instance-specific bound, as opposed to a generic bound, see \cite{MahajanRamanSikdar09} for a discussion on this issue.
Another example of such parameterizations is the problem {\sc Vertex Cover} parameterized above the maximum matching of the given graph. See \cite{ipec11,esa11} for recent results on this problem.

We complete this paper in Section \ref{sec:disc} by stating an open problem on permutation constraint satisfaction problems parameterized above average. \\

{\bf Basics on Parameterized Complexity.} A parameterized problem $\Pi$ can be considered as a set of pairs
$(I,k)$ where $I$ is the \emph{problem instance} and $k$ (usually a nonnegative
integer) is the \emph{parameter}.  $\Pi$ is called
\emph{fixed-parameter tractable (fpt)} if membership of $(I,k)$ in
$\Pi$ can be decided by an algorithm of runtime $O(f(k)|I|^c)$, where $|I|$ is the size
of $I$, $f(k)$ is an arbitrary function of the
parameter $k$ only, and $c$ is a constant
independent from $k$ and $I$. Such an algorithm is called an {\em fpt} algorithm.
Let $\Pi$ and $\Pi'$ be parameterized
problems with parameters $k$ and $k'$, respectively. An
\emph{fpt-reduction $R$ from $\Pi$ to $\Pi'$} is a many-to-one
transformation from $\Pi$ to $\Pi'$, such that (i) $(I,k)\in \Pi$ if
and only if $(I',k')\in \Pi'$ with $k'\le g(k)$ for a fixed
function $g$, and (ii) $R$ is of complexity
$O(f(k)|I|^c)$.

$\Pi$ belongs to XP  if membership of $(I,k)$ in
$\Pi$ can be decided by an algorithm of runtime $|I|^{O(f(k))}$, where $f(k)$ is an arbitrary function of the
parameter $k$ only.  It is well-known that FPT is a proper subset of
XP~\cite{DowneyFellows99,FlumGrohe06,Niedermeier06}.

$\Pi$ is in \emph{para-NP} if membership of $(I,k)$ in
$\Pi$ can be decided in nondeterministic time $O(f(k)|I|^c)$,
where $|I|$ is the size of $I$, $f(k)$ is an arbitrary function of the
parameter $k$ only,
and $c$ is a constant independent from $k$ and $I$. Here,
nondeterministic time means that we can use nondeterministic
Turing machine. A parameterized problem $\Pi'$ is {\em
para-NP-complete} if it is in para-NP and for any parameterized
problem $\Pi$ in para-NP there is an fpt-reduction from $\Pi$ to
$\Pi'$. It is well-known that a
parameterized problem $\Pi$ belonging to para-NP is para-NP-complete if we can reduce an
NP-complete problem to the subproblem of $\Pi$ when the parameter is equal to some
constant \cite{FlumGrohe06}.

For example, consider the {\sc $k$-Colorability} problem, where given a graph $G$ and a
positive integer $k$ ($k$ is the parameter), we are to decide whether
$G$ is $k$-colorable. Since the (unparameterized) {\sc Colorability} problem is in NP,
{\sc $k$-Colorability} is in para-NP. {\sc $k$-Colorability} is para-NP-complete since
{\sc $3$-Colorability} is NP-complete.

For further background and terminology on parameterized algorithms and complexity we
refer the reader to the monographs~\cite{DowneyFellows99,FlumGrohe06,Niedermeier06}.

For an integer $n$, $[n]$ stands for $\{1,\ldots ,n\}.$

\section{Hardness Results}\label{sec:neg}
In this section we give our hardness results. For our results we need the following problem as a starting point for our reductions.

\begin{quote}
\textsc{Linear-3-Sat} \\ \nopagebreak
  \emph{Instance:} A $3$-CNF formula $F$ with clauses $c_1,\ldots , c_m$, and variables $x_1,\ldots ,x_n$ such that $m\leq cn$ for some fixed  constant $c$. That is, number of clauses in $F$ is linear in the number of variables.\\
 \emph{Question:} Decide whether there is a truth assignment satisfying $F$.
 \end{quote}

It is well known that \textsc{Linear-3-Sat} is NP-complete. For example, the well-known theorem of Tovey \cite{Tov84} states that the {\sc 3-SAT}  problem is NP-complete even when the input consists of 3-CNF formula with every variable contained in at most four clauses.

\begin{theorem}\label{th1}
 \MrSAA{} is para-NP-complete for $r(n)=\lceil \log n\rceil$.
\end{theorem}
\begin{proof}
\MrSAA{} is in para-NP as given a truth assignment for an instance $\Phi$ of \MrSAA{}, we can decide, in polynomial time, whether the assignment satisfies at least ${\rm asat}(\Phi)+k$ clauses. To complete our proof of para-NP-completeness, we give a reduction from {\sc Linear-3-SAT} to \MrSAA{} with $k=2$.



Consider a {\sc 3-SAT} formula $F$ with $n$ variables, $x_1,\ldots, x_n$, and $m$ distinct clauses $c_1,\ldots, c_m$. Since $F$ is an input to
\textsc{Linear-3-Sat}, we may assume that $m \le cn$ for some positive constant $c$.

We form a \MrSAA{} instance $F'$ with $n'=2cn$ variables,  the existing variables $x_1,\ldots, x_n$, together with new variables $y_1,\ldots,y_{n'-n}$,
and $m'=2^{\lceil \log n'\rceil + 1}$ clauses. The set of clauses of $F'$ consists of three sets, $C_1$, $C_2$ and $C_3$, described below:

\begin{itemize}
  \item $C_1$ is the complete set of clauses on variables $y_1,\ldots,y_{\lceil \log n'\rceil}$ without the clause consisting of all negative literals, $\bar{c}=(\bar{y_1}\vee \bar{y_2}\vee \ldots \vee \bar{y}_{\lceil \log n' \rceil})$.
  \item $C_2=\{ c_i \vee \bar{y_4}\vee \ldots \vee\bar{y}_{\lceil \log n'\rceil}:\ i\in [m]\}$.
  \item $C_3$ is a set of $m'-|C_1|-|C_2|$  clauses on the variables $y_{\lceil \log n' \rceil+1},\ldots , y_{n'-n}$, of length $\lceil \log n'\rceil$ such that each variable appears in at least one clause and every clause consists of only positive literals.
\end{itemize}

We claim that $F$ is satisfiable if and only if $F'$ is a {\sc Yes}-instance of \MrSAA{} for $k=2$, thus completing the proof.
Since, in $F'$, each clause is of length $\lceil \log n'\rceil$, we have ${\rm asat}(F')=(1-2/m')m'=m'-2.$
Thus, \MrSAA{} for $k=2$ asks whether all the clauses of $F'$ can be satisfied.

Suppose $F$ is satisfied by a truth assignment $x^0.$ Extend this assignment to the variables of  $F'$ by assigning all $y_i$ to be {\sc true}. Since $F$ is satisfied, all the clauses in $C_2$ are satisfied. Every clause in $C_1$ and $C_3$ contains at least one positive literal, and so is satisfied. Hence $F'$ is satisfied.

If $F'$ is satisfied, then $y_1,\ldots,y_{\lceil \log n'\rceil}$ must all be set to {\sc true} (otherwise, there is a clause in $C_1$ that is not satisfied). As a result, the set $C_2$ of clauses can be simplified to the {\sc 3-SAT} formula $F$, and thus $F$ must be satisfied.
\end{proof}


It is not hard to prove Theorem \ref{th1} without starting from {\sc Linear-3-SAT}. We use {\sc Linear-3-SAT} to ensure that $n'=O(n)$ which is necessary in the proof of the next theorem. Hereafter, we will assume the Exponential Time Hypothesis (ETH), which is stated below.
\begin{quote}
Exponential Time Hypothesis (ETH)~\cite{RussellP01}: There is a positive real $s$ such that 3-SAT cannot be solved in time $2^{sn}n^{O(1)}$.
Here $n$ is the number of variables.
\end{quote}
Using the sparsification  lemma~\cite[Corollary 1]{RussellP01}, one may assume that in the input formula $F$ to {\sc 3-SAT},
every variable occurs in at most $p$ clauses for some positive constant $p$.  For completeness we sketch the proof here.
\begin{lemma}[\cite{RussellP01}] {\rm \bf (Sparsification Lemma)}
For every $\ep > 0$ and positive integer $r$, there is a constant $C$
so that any $r$-CNF formula $F$ with $n$ variables, can be expressed as
$F=\vee_{i=1}^t Y_i$, where $t \leq 2^{\ep n}$ and each $Y_i$ is an $r$-CNF formula
with at most $Cn$ clauses. Moreover, this disjunction can be computed by an algorithm
running in time $2^{\ep n}n^{O(1)}$.
\end{lemma}

\begin{proposition}[\cite{RussellP01}]
Assuming the ETH, there is a positive real $s'$ such that Linear 3-SAT cannot be solved in time $O(2^{s'n})$.
\end{proposition}
\begin{proof}
Suppose the proposition does not hold. That is, for all positive real $s'$ {\sc Linear 3-SAT} can be solved in time $O(2^{s'n})$.
Now consider a 3-CNF formula $F$.  Using the algorithm for {\sc Linear 3-SAT} we will show that for every positive real $c$,
3-SAT can be solved in time $2^{cn}n^{O(1)}$, contradicting the ETH.
Let ${\ep}'<c/2.$ Using Sparsification Lemma, we produce linear $3$-CNF formulas
$Y_1,\ldots ,Y_t, t\leq 2^{{\ep}' n}$, in time $2^{{\ep}' n}n^{O(1)}$. We can solve
all of them in time $2^{2{\ep}' n} n^{O(1)}=2^{cn}n^{O(1)}$ and so obtain a solution for $F$ in time
$O(2^{cn})$, a contradiction.
\end{proof}
Thus, it shows that  {\sc Linear-3-SAT} cannot be solved in time $2^{o(n)}$ unless ETH fails.
Using the ETH, we strengthen Theorem \ref{th1}.

\begin{theorem}\label{th2}
Assuming the ETH, \MrSAA{} is not in XP for any $r(n) \ge \log\log n+\phi(n)$, where $\phi(n)$ is any unbounded strictly increasing function of $n$.
\end{theorem}
\begin{proof}
Let $\phi(n)$ be an unbounded strictly increasing function of $n$.
Note that if the theorem holds for some unbounded strictly increasing function $\psi(n)$ of $n$, it holds also for any strictly increasing function $\psi'(n)$ such that $\psi'(n)\ge \psi(n)$ for every $n\ge 1.$
Thus, we may assume that $\phi(n)\le \log\log n$.

Consider a {\sc 3-SAT} instance $I$ with a linear number of clauses, and reduce it to a \textsc{Max-$r(n')$-Sat-AA} instance with $r(n')=\lceil \log n'\rceil$ as in the proof of Theorem \ref{th1}. Note that $I$ has $O(n')$ variables.
Let $F'$ be the formula of the \textsc{Max-$r(n')$-Sat-AA} instance and let $n$ be the maximum integer such that $\log n'\ge \log\log n + \phi(n)$. Add $n-n'$ new variables to $F'$ together with a pair of contradicting unit clauses, $(x)$, $(\bar{x})$ for each new variable. Let $F$ denote the resulting formula.
The total number $n$ of variables in $F$ is such that $r(n)=\lceil \log n' \rceil\ge \log\log n + \phi(n)$. Note that $$n\le 2^{n'/{2^{\phi(n)}}}\le 2^{n'/{2^{\phi(n')}}}= 2^{o(n')},$$ and hence it takes $2^{o(n')}$ time to construct $F$.

Observe that for any value of $k$, $F'$ is a {\sc Yes}-instance of \textsc{Max-$r(n')$-Sat-AA} if and only if $F$ is a {\sc Yes}-instance of \textsc{Max-$r(n)$-Sat-AA}. We established in the proof of Theorem \ref{th1} that $F'$ is a {\sc Yes}-instance of \textsc{Max-$r(n')$-Sat-AA} for $k=2$  if and only if the {\sc 3-SAT} instance is satisfiable. Thus, $F$ is a {\sc Yes}-instance of \textsc{Max-$r(n)$-Sat-AA} for $k=2$  if and only if the {\sc 3-SAT} instance $I$ is satisfiable. Therefore, if there was an XP algorithm for \MrSAA{}, then for $k=2$ it would have running time $n^{O(1)}=2^{o(n')}$, contradicting the ETH.
\end{proof}

\section{Algorithmic Results}\label{sec:pos}

To prove the main result of this section, Theorem \ref{thm:maxsataafpt}, we reduce \MrSAA{} to \textsc{Max-$r(n)$-Lin2-AA} defined below.

In the problem \textsc{MaxLin2}, we are given a system $S$ consisting of
$m$ equations in $n$ variables, where each equation is of the form $\prod_{i \in I}x_i = b$,
$b \in \{-1, 1\}$ and each variable $x_i$ may only take a value from $\{-1, 1\}$.
Each equation is assigned a positive integral weight and we wish to find an assignment of
values to the variables in order to maximize the total weight of satisfied equations.

Let $W$ be the sum of the weights of all equations in $S$
and let ${\rm sat}(S)$ be the maximum total weight of equations that can be satisfied simultaneously.
Note that $W/2$ is a tight lower bound on ${\rm sat}(S)$. Indeed, consider choosing assignments to the
variables uniformly at random. Then $W/2$ is the expected weight of satisfied equations
(as the probability of each equation to
be satisfied is $1/2$) and so is a lower bound; to see
the tightness consider a system consisting of pairs of equations
of the form $\prod_{i\in I}x_i=-1,\ \prod_{i\in I}x_i=1$ of weight 1, for some non-empty $I \subseteq [n]$.
This leads to the following parameterized problem:


\begin{quote}
  {\sc Max-$r(n)$-Lin2-AA}\\ \nopagebreak
    \emph{Instance:} A system $S$ of equations $\prod_{i \in I_j}x_i = b_j$, where $b_j \in \{-1,1\}$, $|I_j|\le r(n)$, $x_i \in \{-1, 1\}$, and $j\in [m]$,
in which Equation $j$ is assigned a positive integral weight $w_j$, and a nonnegative integer $k$.
Let $W=\sum_{j=1}^m w_j.$\\ \nopagebreak
    \emph{Parameter:} $k$.\\ \nopagebreak
    \emph{Question:} Decide whether ${\rm sat}(S)\ge W/2+k$. \nopagebreak
  \end{quote}

The \emph{excess} for $x^0=(x^0_1,\ldots ,x^0_n)\in \{-1, 1\}^n$ over $S$ is $$\varepsilon_S(x^0)=\frac{1}{2}\left[\sum_{j=1}^m c_j\prod_{i\in I_j}x^0_i\right],$$ where $c_j=w_{j}b_j$. Observe that $\varepsilon_S(x^0)$ is the difference between the total weight of equations satisfied by $x^0$ and $W/2$. Thus, the answer to {\sc MaxLin2-AA} is {\sc Yes} if and only if
$\varepsilon_S(x^0)\ge k$ for some $x^0$.

Consider two reduction rules for {\sc MaxLin2} introduced in \cite{GutKimSzeYeo11}.

  \begin{krule}\label{rulerank}
Let $A$ be the matrix over $\mathbb{F}_2$ corresponding to the set of equations in $S$, such that $a_{ji} = 1$ if variable $x_i$ appears in equation $e_j$, and $0$ otherwise.
  Let $t={\rm rank} A$ and suppose columns $a^{i_1},\ldots ,a^{i_t}$ of $A$ are linearly independent.
  Then delete all variables not in $\{x_{i_1},\ldots ,x_{i_t}\}$ from the equations of $S$.
  \end{krule}

  \begin{krule}\label{rule1}
  If we have, for a subset $I$ of $[n]$, an equation $\prod_{i \in I} x_i =b_I'$
  with weight $w_I'$, and an equation $\prod_{i \in I} x_i =b_I''$ with weight $w_I''$,
  then we replace this pair by one of these equations with weight $w_I'+w_I''$ if $b_I'=b_I''$ and, otherwise, by
  the equation whose weight is bigger, modifying its
  new weight to be the difference of the two old ones. If the resulting weight
  is~0, we delete the equation from the system.
  \end{krule}

The two reduction rules are of interest due to the following:

  \begin{lemma}\cite{GutKimSzeYeo11}
  Let $S'$ be obtained from $S$ by Rule~\ref{rulerank} or \ref{rule1}.
  Then the maximum excess of $S'$  is equal to the maximum excess of $S$.
  Moreover, $S'$ can be obtained from $S$ in time polynomial in $n$ and $m$.
  \end{lemma}

Using techniques from linear algebra, the authors of \cite{CroFellowsEtAl11} showed the following:

\begin{theorem}\label{lemYes}
Let $J$ be an instance of {\sc Max-$r(n)$-Lin2-AA} in which system $S$ is reduced with respect to Rules \ref{rulerank} and \ref{rule1}.
If $n\ge (2k-1)r(n)+1$ then the answer to $J$ is {\sc Yes}.
\end{theorem}

Let $I$ be an instance of {\sc MaxSat-AA} given by a CNF formula $F$ with clauses $c_1,\ldots , c_m$, and variables $x_1,\ldots ,x_n.$ It will be convenient for us to denote {\sc true} and {\sc false} by $-1$ and $1$, respectively. For a truth assignment $x^0\in \{-1,1\}^n$, the {\em excess} $\ep_I(x^0)$ for $x^0$ is the number of clauses satisfied by $x^0$ minus ${\rm asat}(x^0)$. Thus, the answer to $I$ is {\sc Yes} if and only if there is an assignment $x^0$ with $\ep_I(x^0)\ge k.$

{\sc Max-$r(n)$-Sat-AA} is related to {\sc Max-$r(n)$-Lin2-AA} as follows. Results similar to Lemma \ref{lemLinSat} have been proved in \cite{AloGutKimSzeYeo11,CroGutJonKimRuz10}.

\begin{lemma}\label{lemLinSat}
Let $I$ be an instance of {\sc Max-$r(n)$-SAT-AA} with $n$ variables, $m$ clauses and parameter $k$. Then in time $2^{r(n)}m^{O(1)}$ we can produce an instance $J$ of {\sc Max-$r(n)$-Lin2-AA}  with parameter $k 2^{r(n)-1}$ such that $I$ is a {\sc Yes}-instance if and only if $J$ is a {\sc Yes}-instance and
$J$ is reduced by Rule \ref{rule1}. Moreover, for any truth assignment $x\in \{-1,1\}^n$, $\ep_J(x) = \varepsilon_I(x) \cdot 2^{r(n)-1}$.
\end{lemma}
\begin{proof}
Let $I$ be an instance of {\sc Max-$r(n)$-SAT-AA} with clauses $c_1,\ldots , c_m$ and variables $x_1,\ldots , x_n.$
For a clause $c_j$, $\var(c_j)$ will denote the set of variables in $c_j$ and $r_j$ the number of literals in $c_j.$
For every $j\in [m]$, let $$h_j(x)=2^{r(n)-r_j}[1-\prod_{x_i\in \var(c_j)}(1+d_{ij}x_i)],$$
where $d_{ij}=1$ if $x_i$
is in~$c_j$ and $d_{ij}=-1$ if $\bar{x}_i$
is in~$c_j$.

Let $H(x)=\sum_{j=1}^mh_j(x).$ We will prove that for a truth assignment $x\in \{-1,1\}^n$, we have \begin{equation}\label{eq1}H(x)=2^{r(n)}\ep_I(x).\end{equation}

Let $q_j=1$ if $c_j$ is satisfied by $x$ and $q_j=0$, otherwise. Observe that $h_j(x)/(2^{r(n)-r_j})$ equals $1-2^{r_j}$ if $q_j=0$ and 1, otherwise.
Thus, \begin{center}
$\begin{array}{rcl}
  H(x) & = & \sum^m_{j=1}[2^{r(n)-r_j}q_j+(2^{r(n)-r_j}-2^{r(n)})(1-q_j)] \\
  & = & 2^{r(n)}[\sum^m_{j=1}q_j-\sum_{j=1}^m(1-2^{-r_j})] \\
  & = &2^{r(n)}\ep_I(x). \\
  \end{array}$
\end{center}

It follows from (\ref{eq1}) that the answer to $I$ is {\sc Yes}
if and only if there is a truth assignment $x$ such that
\begin{equation}\label{eq2} H(x)\ge k2^{r(n)}.\end{equation}

Algebraic simplification of $H(x)$ will lead us to the following:

\begin{equation}\label{foureq} H(x)=\sum_{S\in {\cal F}}c_S\prod_{i\in S}x_i,\end{equation}
where ${\cal F}=\{\emptyset\neq S\subseteq [n]:\ c_S\neq 0, |S|\le r(n) \}$. The simplification can be done
in time $2^{r(n)}m^{O(1)}$.

Observe that by replacing each term $c_S\prod_{i\in S}x_i$ with the equation $\prod_{i\in S}x_i = 1$ if $c_S \ge 0$ and $\prod_{i\in S}x_i = -1$ if $c_S < 0$, with weight $|c_S|$,
the sum $\sum_{S\in {\cal F}}c_S\prod_{i\in S}x_i$ can be viewed as twice the excess of an instance $J$ of
{\sc Max-$r(n)$-Lin2-AA}. Let $k 2^{r(n)-1}$ be the parameter of $J$. Then, by (\ref{eq2}), $I$ and $J$ are equivalent.

Note that the algebraic simplification of $H(x)$ ensures that $J$ is reduced by Rule \ref{rule1}.
This completes the proof.
\end{proof}

It is important to note that the resulting instance $J$ of {\sc Max-$r(n)$-Lin2-AA} is not necessarily reduced under Rule \ref{rulerank} and, thus, reduction of $J$ by Rule \ref{rulerank} may result in less than $n$ variables.

{From} Theorem \ref{lemYes} and Lemma \ref{lemLinSat} we have the following fixed-parameter-tractability result for {\sc Max-$r(n)$-SAT-AA}.

\begin{theorem}\label{thm:maxsataafpt}
 \MrSAA{} is (i) in XP for $r(n)\le \log\log n-\log\log\log n$ and (ii) fixed-parameter tractable for $r(n)\le \log\log n-\log\log\log n-\phi(n)$, for any unbounded strictly increasing function $\phi(n)$.
\end{theorem}

\begin{proof}
We start by proving Part (ii). Let $\phi(n)$ be an unbounded strictly increasing function of positive integral argument. Note that $\phi(n)$ can be extended to a continuous positive strictly increasing function $\phi(t)$ of real argument $t\ge 1$. Thus, $\phi(t)$ has an inverse function $\phi^{-1}(t).$
We may assume that $\phi(n)\ge 0$ for each $n\ge 1$ as otherwise we may consider only $n$ large enough.

Let $r(n)\le \log\log n-\log\log\log n-\phi(n)$ and consider an instance $I$ of {\sc Max-$r(n)$-Sat-AA}.
Note that $2^{r(n)} \le n$. Therefore by Lemma \ref{lemLinSat}, in polynomial time we can reduce $I$ into an instance $J$ of {\sc Max-$r(n)$-Lin2-AA}, such that $I$ is a {\sc Yes}-instance if and only if $J$ is a {\sc Yes}-instance with parameter $k \cdot 2^{r(n)-1}$.
Consider the {\sc Max-$r(n)$-Lin2-AA} instance $J'$ with $n'$ variables formed by reducing $J$ by Rule \ref{rulerank}.
If $n'\le \log n$, $J'$ may be solved in polynomial time by trying all $2^{n'}\le n$ assignments to the variables of $J'$. Thus, we may assume that $n'> \log n.$

If $n' \ge (k2^{r(n)}-1) r(n)+1$, then by Theorem~\ref{lemYes} and Lemma~\ref{lemLinSat},  $I$ is a {\sc Yes}-instance. Thus, it remains to consider the case
$\log n < n' \le (k2^{r(n)}-1) r(n)$. We have $$\log n \le (k2^{r(n)}-1) r(n) \mbox{ and so } \log n \le k (\log\log n) \cdot \log n/(2^{\phi(n)}\log\log n).$$ This simplifies to $\phi(n) \le \log k$ and so $n\le \phi^{-1}(\log k).$
Hence, $I$ can be solved in time $2^{\phi^{-1}(\log k)}m^{O(1)}$ by trying all possible assignments to variables
of $J'$.

Now we will prove Part (i). Let $r(n)\le \log\log n-\log\log\log n$ and consider an instance $I$ of {\sc Max-$r(n)$-Sat-AA}.
As in Part (ii) proof, we reduce $I$ into an instance $J$ of {\sc Max-$r(n)$-Lin2-AA}, such that $I$ is a {\sc Yes}-instance if and only if $J$ is a {\sc Yes}-instance with parameter $k \cdot 2^{r(n)-1}$.
Consider the {\sc Max-$r(n)$-Lin2-AA} instance $J'$ with $n'$ variables formed by reducing $J$ by Rule \ref{rulerank}.
If $n'\le k\log n$, $J'$ may be solved in XP time by trying all $2^{n'}\le n^k$ assignments to the variables of $J'$. Thus, we may assume that $n'> k\log n.$
If $n' \ge (k2^{r(n)}-1) r(n)+1$, then by Theorem~\ref{lemYes} and Lemma~\ref{lemLinSat},  $I$ is a {\sc Yes}-instance. Thus, it remains to consider the case
$k\log n < n' \le (k2^{r(n)}-1) r(n)$. We have $k\log n \le kr(n)2^{r(n)}$ and so $\log\log n \le \log\log n-\log\log\log n,$ a contradiction. Thus, this case is impossible and we can solve $I$ in XP time.
\end{proof}

\section{Open Problems}\label{sec:disc}

It would be interesting to close the gap between the inequalities of Theorems \ref{th2} and \ref{thm:maxsataafpt}.

Apart from {\sc MaxLin-AA} and {\sc MaxSat-AA} mentioned above, there are some other constraint satisfaction problems parameterized above a tight lower bound whose complexity has been established in the last few years. One example is {\sc $r$-Linear-Ordering-AA} for $r=2$ and 3. Let $r\ge 2$ be a fixed integer. In {\sc $r$-Linear-Ordering}, given a positive integer $n$ and a multiset $\cal C$ of $r$-tuples of distinct elements from $[n]$, we wish to find a permutation $\pi:\ [n]\rightarrow [n]$ which maximizes that the number of {\em satisfied} $r$-tuples, i.e., $r$-tuples $(i_1,i_2,\ldots ,i_r)$ such that $\pi(i_1)<\pi(i_2)<\cdots <\pi(i_r).$ Let $m$ stand for the number of $r$-tuples in $\cal C$.

Let $\tau:\ [n]\rightarrow [n]$ be a random permutation (chosen uniformly from the set of all permutations). Observe that the probability that an $r$-tuple is  satisfied by $\tau$ is $1/r!$. Thus, by linearity of expectation, the expected number of $r$-tuples satisfied by $\tau$ is $m/r!$. Using conditional expectation derandomization method \cite{alon}, it is not difficult to obtain a polynomial time algorithm for finding a permutation $\pi$ which satisfies at least $m/r!$ $r$-tuples. Thus, we can easily obtain an $1/r!$-approximation algorithm for {\sc $r$-Linear Ordering}. It is remarkable that for any positive $\ep$ there  exists no polynomial $(1/r! + \ep )$-approximation algorithm provided the Unique Games Conjecture (UGC) of Khot \cite{khot} holds. This result was proved by Guruswami {\em et al.} \cite{GurManRag} for $r=2$, Charikar {\em et al.} \cite{ChaGurMan} for $r=3$ and, finally,  by Guruswami {\em et al.} \cite{GurHasManRagCha} for any $r$.

Observe that every permutation $\pi$ satisfies exactly one $r$-tuple in the set  $\{(i_1,i_2,\ldots , i_r):\ \{i_1,i_2,\ldots ,i_r\}=[r]\}$ and, thus,
$m/r!$ is a tight lower bound on the maximum number of $r$-tuples that can be satisfied by a permutation $\pi$. It is natural to ask what is the parameterized complexity of the following problem {\sc $r$-Linear-Ordering-AA}: decide whether there is a permutation $\pi$ which satisfies at least $m/r! + k$ $r$-tuples, where $k\ge 0$ is the parameter. Gutin {\em et al.} \cite{GutKimSzeYeo11} and \cite{GutinIerselMnichYeo} proved that {\sc $r$-Linear-Ordering-AA} is fixed-parameter tractable for $r=2$ and $r=3$, respectively. The complexity of {\sc $r$-Linear-Ordering-AA} for $r\ge 4$ remains an open problem \cite{GutinIerselMnichYeo}. Note that
if {\sc $r$-Linear-Ordering-AA} is fixed-parameter tractable for some $r$, then all permutation constraint satisfaction problems of arity $r$ parameterized above average are fixed-parameter tractable too (see \cite{GutinIerselMnichYeo} for the definition of a permutation constraint satisfaction problem of arity $r$ and a proof of the above-mentioned fact).

\medskip

\paragraph{Acknowledgment}
This research was partially supported by an International Joint grant of Royal Society.

\urlstyle{rm}


\end{document}